\documentclass{article}

\usepackage{natbib}
\usepackage[final]{neurips_2022}

\usepackage[utf8]{inputenc} % allow utf-8 input
\usepackage[T1]{fontenc}    % use 8-bit T1 fonts
\usepackage{hyperref}       % hyperlinks
\usepackage{url}            % simple URL typesetting
\usepackage{booktabs}       % professional-quality tables
\usepackage{amsfonts}       % blackboard math symbols
\usepackage{nicefrac}       % compact symbols for 1/2, etc.
\usepackage{microtype}      % microtypography
\usepackage{xcolor}         % colors

\usepackage{graphicx}

\usepackage[ruled,linesnumbered, noend]{algorithm2e}
\usepackage{amssymb}
\usepackage{amsthm}
\usepackage{amsmath}
\usepackage{svg}
\usepackage{multirow}

\usepackage{subfig}

\newtheorem{proposition}{Proposition}
\newtheorem{theorem}[proposition]{Theorem}
\newtheorem{lemma}[proposition]{Lemma}
\newtheorem{corollary}[proposition]{Corollary}

\newcommand\argmin{\arg\min}
\newcommand\argmax{\arg\max}

\newcommand{\guess}{\tau}
\newcommand{\eps}{\epsilon}

\newcommand{\OPT}{{\sf OPT}}

\title{Fair Range $k$-center}

\author{Matthew Jones \\
    Khoury College \\
    jones.m@northeastern.edu \\
    \And
    Huy L\^{e} Nguy\~{\^{e}}n \\
    Khoury College \\
    hu.nguyen@northeastern.edu \\
    \And
    Thy Nguyen \\ 
    Khoury College \\
    nguyen.thy2@northeastern.edu \\}
\begin{document}

\maketitle

\begin{abstract}
We study the problem of fairness in k-centers clustering on data with disjoint demographic groups. Specifically, this work proposes a variant of fairness which restricts each group's number of centers with both a lower bound (minority-protection) and an upper bound (restricted-domination), and provides both an offline and one-pass streaming algorithm for the problem. In the special case where the lower bound and the upper bound is the same, our offline algorithm preserves the same time complexity and approximation factor with the previous state-of-the-art. Furthermore, our one-pass streaming algorithm improves on approximation factor, running time and space complexity in this special case compared to previous works. Specifically, the approximation factor of our algorithm is 13 compared to the previous 17-approximation algorithm, and the previous algorithms' time complexities have dependence on the metric space's aspect ratio, which can be arbitrarily large, whereas our algorithm's running time does not depend on the aspect ratio.
\end{abstract}

\section{Introduction}
%Machine learning algorithm are increasingly responsible for automating the decision-making process in high-stake scenarios such as in  healthcare, education and job searching \cite{FairRef}. 
% While machine learning algorithms are designed to optimize classical measure of efficiency and utility, they are vulnerable to biases  and exhibit unfair behavior patterns \cite{tambe2019artificial,cowgill2018bias,datta2015automated, sweeney2013discrimination}. In the context of clustering, where a small subset of a dataset is chosen to be approximately representative of the whole dataset, evidence of this bias is also observable: a Google Images search for "CEO" returns an image set with a significantly higher proportion of men than the proportion of male CEOs \cite{ImageBias}. 
%To address these underlying biases, fair machine learning algorithms are designed to correct the issue by explicitly restricting the ability of algorithms to produce unfair behavior.

Fairness has been studied in a number of clustering formulations.
In this work, we focus on the setting of $k$-center, where the input points come from $m$ disjoint demographic groups and the goal is to select $k$ representatives so as to minimize the $k$-center objective. In an ideal situation, the proportion of points from a demographic group would be nearly equal in the entire dataset and the representative set (the $k$ centers). $k$-center clustering has been studied with this equality-based fairness \cite{FairkCCfDS,nguyen2020fair,CKR20}, where the most efficient algorithm is a $3$-approximation in $O(nk)$ time \cite{nguyen2020fair}. These approaches take in the required numbers of centers $m_i$ from each demographic group $i$ as input, and output $k$ centers with the exact demographic constraints. The algorithms are based on polynomial time algorithms for matroid centers due to \cite{chen2016matroid} where the centers are required to form an independent set in a given matroid, which is a uniform matroid in the classic $k$ centers and a generalized partition matroid in the fairness setting. In the streaming setting, \cite{kale2019small} uses techniques of \cite{guha2009tight} and \cite{chen2016matroid} to give $(17+\epsilon)$-approximation one-pass and $(3+\epsilon)$-approximation two-pass algorithms with running time $O_{\epsilon}((nk + k^{3.5}) + k^2 \log (\Lambda))$, where $k$ is the rank of the matroid, $\Lambda$ is the aspect ratio of the metric, and $\epsilon$ terms are hidden by the $O_{\epsilon}$ notation. 

Although equality-fairness is ideal for maintaining fairness with respect to all demographic groups, it is still an open question of how one can efficiently make a trade-off between fairness and utility in the fair $k$-center problem. One approach is to slightly adjust the demographic constraints $m_i$'s to include more points from groups that help optimizing the $k$-center distance. For a fixed budget of $k$ centers, this also implies reduction of the number of required centers from other demographic groups. Modifying the constraints $m_i$'s  would allow one to still utilize the approaches of \cite{FairkCCfDS,nguyen2020fair,CKR20} while trading fairness for improved utility. As we will show in the experimental section, this heuristic, though seemingly intuitive, is ineffective; it is non-trivial to efficiently determine an allocation of $k$ centers from the  demographic groups that would benefit the objective.  

The limitation on controlling the fairness-utility trade-off of the previous approaches motivates our work on a novel fairness constraint where the number of representatives from each demographic group $i$ is restricted not to an exact number but a range from at least $l_i = \frac{\alpha_i \left\lvert  S_i \right\rvert}{n} k $ centers to at most $u_i= \frac{ \beta_i \left\lvert  S_i \right\rvert}{n}k$ centers, for $\alpha \in (0,1]$ and $\beta \in [1,\frac nk]$ . Within the context of the fairness-utility trade-off, a smaller value of $\alpha_i$ (larger value of $\beta_i$) allows the demographic make-up of the $k$ centers to deviate from the fairest clustering and obtain improved utility. In practice, a canonical example of the range-based fairness constraint is the well-known four-fifths rule which states that “a selection rate for any race, sex, or ethnic group which is less than four-fifths (or $80\%$) of the rate for the group with the highest rate will generally be regarded by the Federal enforcement agencies as evidence of adverse impact” ~\cite{eeoc78}. 

Formally, the range-based fair $k$ centers problem takes as input the value $k$, the dataset $S$ of size $n$, the bounds $l_i$ and $u_i$ for each demographic group $i$, and a metric $d$. We use $S_i$ to denote the subset of $S$ with demographic group $i$. The goal of the problem is to find the subset $C$ of $S$ of size $k$ that satisfies the range-based fairness constraint and minimizes the $k$-center objective. Formally, we want to find
\begin{align*}
    {\argmin}_{\substack{C = \{c_1,c_2, ...,c_k\} \subseteq S\\\forall i\in [m] : l_i \le |C\cap S_i| \le u_i}} \hspace{0.05cm} \max_{s \in S} \min_{c \in C} d(s,c).
\end{align*}
The bounds on $\{\alpha_i\}_{i \in m}, \{\beta_i\}_{i \in m}$  imply
 $\sum_{i=1}^m l_i \le k \le \sum_{i=1}^m u_i$ and  $\forall i\in [m] : 0 \le l_i \le u_i \le |S_i|$, so a feasible solution must exist.

Despite its better alignment with practice, the new formulation poses significant technical challenge since it is not a matroid constraint. It is in fact not even a down-closed constraint (a more general constraint also previously studied for clustering) because a subset of a feasible set needs not be feasible due to the lower bounds. Thus, existing algorithms for matroid centers in  \cite{kale2019small, chen2016matroid} do not apply to this formulation and a new approach is needed.

\subsection{Our Contributions}

In this work, we present new algorithms for range-based fairness in both offline and streaming settings with constant approximations and running time $O(nk)$. Note that this running time is the same as the time to compute the distance between all points and a given set of $k$ centers and hence, it is essentially as good as possible.

In the classical setting, our algorithm finds a $3$-approximate solution in $O(nk)$ time. The approximation factor and time complexity match the best known bound in \cite{nguyen2020fair} for the special case with $l_i=u_i~\forall i$, and it is the first solution for the new formulation. Empirically, our experiments show that the algorithm's running time and performance are also competitive with previous algorithms designed for equality-based fairness, which is a special case of range-based fairness.

In the streaming setting, our one-pass algorithm finds a $(13+\eps)$-approximate solution in space $O_{\epsilon}\left(km\right)$ and total time $O_{\epsilon}(nk)$ (or equivalently $O_\epsilon(k)$ amortized time per input point). Despite being more general, our algorithm has better approximation factor, running time and space complexity than previous results~\cite{kale2019small} for equality-based constraint ($l_i=u_i~\forall i$).

\subsection{Related Works}

A different formulation of fair clustering where equality is sought for the members of each cluster as opposed to the centers was introduced to $k$-center clustering in \cite{chierichetti2017fair} for 2 demographic groups, with a 4-approximation algorithm. This idea was generalized to multiple demographic groups in \cite{rosner2018privacy}. The range-based generalization of \cite{chierichetti2017fair} is defined in \cite{Bercea2019OnTC} with a 5-approximation algorithm. The work of \cite{bera2019fair} also uses range-based fairness, giving a $(\rho + 2)$-approximation algorithm for a fair $\ell_p$-norm clustering problem which has an unfair $\rho$-approximation algorithm and also allowing overlapping demographic groups, with scaling additive error. A similar definition of fairness is also used for $k$-means in \cite{schmidt2020}, who also present a streaming algorithm for fair $k$-means, and  \cite{kleindessner2019guarantees} considers this fairness within the spectral clustering framework.

Another line of work on fair clustering is \cite{celis2018fair}. Their paper studies the problem of fair summarization with equality-based fairness, but their objective is instead to maximize the diversity score of the selected points. A different fairness definition based on diversity is in \cite{withDiversity2010}, where each cluster is individually required to be sufficiently diverse. Another examination of fairness in \cite{chenProportionally} does not rely on demographic groups, but instead concerns itself with cluster size (in assigned points) by allowing sets of points with size $n/k$ the option to form a cluster. This is similar but distinct from capacitated clustering, where the number of points assigned to any cluster center is bounded above. This clustering has a 6-approximation algorithm for capacitated $k$-center \cite{Khuller2000TheCK} and a $(7+\epsilon)$-approximation FPT algorithm for capacitated $k$-medians \cite{Adamczyk2018ConstantFF}. As a more general form of clustering which could potentially be used to model a fairness constraint, \cite{Chakrabarty2018GeneralizedCP} presents an algorithm for clustering with outliers which chooses the set of centers as an item in a down-closed family $\mathcal{F}$ of subsets of $S$, where the contents of $\mathcal{F}$ limits the set of centers chosen in the clustering. Down-closed constraint is a generalization of a matroid constraint but it does not capture a ranged-based notion.

\section{Preliminaries}
\label{sec:prelim}

\subsection{Gonzalez's 2-approximation for $k$-center}
\label{sec:gonz}

The (unfair) $k$-center problem is known to be NP-hard to solve or approximate within a factor less than 2. Furthermore, there exists an algorithm which takes $O(nk)$ time and yields a 2-approximation for the $k$-center problem \cite{Gonzalez}. 

This algorithm, commonly referred to as the Gonzalez algorithm, involves selecting the first center arbitrarily and then iteratively choosing the remaining $k-1$ centers. Each center is chosen as the point in $S$ which is farthest from all previously chosen centers, i.e. the $i$th chosen center would be the $s \in S$ which was responsible for the $k$-center objective value on the first $i-1$ chosen centers. Let $a_1,...,a_k$ denote the sequence centers returned by an instance of the Gonzalez algorithm in the order which they were chosen. Also, let $d_i$ for $i \in \{2,...,k\}$ be the minimum distance between $a_i$ and the set $\{a_1,...,a_{i-1}\}$. An important and well-known lemma is:

\begin{lemma}
The sequence $\left(d_2,d_3,...,d_k\right)$ is non-increasing.\label{seqdec}
\end{lemma}

%\begin{proof}
%\begin{align*}
%d_{i+1} &= d(a_{i+1}, \{a_j\}_{j<i+1})\\
%&= \min\{d(a_{i+1},a_i), d(a_{i+1}, \{a_j\}_{j<i})\}\\
%&\le d(a_{i+1}, \{a_j\}_{j<i})\\
%&\le d(a_i, \{a_j\}_{j<i})\\
%&= d_i
%\end{align*}
%The last inequality holds because $a_i$ was chosen as the $i$th item in the sequence, so it must be farther, or at least no closer, than any point in $S$ to $\{a_j\}_{j<i}$. The whole chain of inequalities yields to $d_{i+1} \le d_i$, proving the lemma inductively.
%\end{proof}

\subsection{Fair k-center Under Equality-Based Fairness}

We briefly review a previous algorithm for $k$-center under the equality-based fairness constraint, which serves as a base for the algorithms in this paper. This algorithm comes from \cite{nguyen2020fair}, and is an improvement on \cite{chen2016matroid}.

The algorithm in \cite{nguyen2020fair} for $k$-center under the equality-based fairness can be divided into four larger steps. First, compute a sequence of unfair centers using the Gonzalez algorithm. Second, compute the largest prefix length $h$ of the sequence such that the first $h$ centers can be substituted by (or "shifted" to) a set of nearby points in $S$ that satisfy the fairness constraint. The points in $S$ are assigned to their nearest centers among the $h$ centers and we are only allowed to substitute a center with a point assigned to it. Third, find the minimum-distance such shift on that prefix. Finally, pick the remaining $(k - h)$ centers arbitrarily to obtain $k$ centers which satisfy fairness.

The "fair-shift problem" is as follows: given a set $S$, a set of centers $A \subseteq S$, a distance value $d'$, and a metric $d:S\times S \to \mathbb{R}$, find an injective mapping $g:A\rightarrow S$ such that
\[\forall a \in A, d(g(a),a) < d'\] and one can choose the remaining $(k-|A|)$ centers arbitrarily up to demographic group such that the set of size $k$ including the arbitrary centers and $g(A)$ satisfies the fairness constraint. A set of points $A$ satisfies the "fair-shift constraint" if the mapping $g$ exists or, equivalently, the set $g(A)$ exists. Solving this problem is sufficient for solving the second and third major steps in the above algorithm, by varying the set $A$ and value $d'$. By the design of the algorithm, {\em we only need to solve the fair-shift problem when $d'$ is sufficiently small that balls of radius $d'$ around $A$ are disjoint, so each $s\in S$ is a substitution candidate for at most one center in $A$}. 

For the equality-based fairness constraint, solving the fair-shift problem is accomplished via maximum matching solved as a maximum-flow problem. The setup matches vertices representing the points in $A$ to vertices representing demographic group centers under the fairness constraint, where each demographic group can match with as many vertices in $A$ as fairness allows. A vertex representing $a \in A$ can match with a vertex representing the demographic group $i$ if and only if $\exists s \in S_i\text{ : } d(a,s) < d'$, where $s = a$ is allowed when $a \in S_i$. Thus, a matching of size $|A|$ exists iff there is a set $g(A)$ which solves the earlier problem. The proof of this fact, as well as the $O(nk)$ time complexity, is seen in Section 5.2 of \cite{nguyen2020fair}.

Our offline algorithm for range-based fair $k$-center in this paper is outlined in algorithm \ref{alg:3approx}, and follows the same general structure as the algorithm in \cite{nguyen2020fair}. Since the fair-shift constraint inherently allows Step 4 to be possible, the significant step to develop algorithm \ref{alg:3approx} lies in formalizing and solving the fair-shift problem for range-based fairness.

\begin{algorithm}[t]
\caption{Outline of the 3-approximation algorithm for $k$-center with range-based fairness}
\label{alg:3approx}
\DontPrintSemicolon
\SetAlgoLined
\LinesNumbered
\KwIn{a set of points $S = \{s_1,...,s_n\}$ each with a demographic group value $f_j$, a distance metric $d$, the values $l_i$ and $u_i$ for each demographic group, a value $k$}
\KwOut{A set $C$ such that $C \subseteq S$ and $l_i \le \lvert\{C \cap S_i\}\rvert \le u_i$ for all demographic group values $f$}
\BlankLine
Compute a sequence of $k$ (unfair) centers using the Gonzalez algorithm. \;

Find the largest integer $h$ such that the first $h$ items in the sequence satisfy the fair-shift constraint with $d'=d_h/2$.\;

Find the set of points which can substitute for the first $h$ items in the sequence to fulfill the fair-shift constraint such that the maximum distance between an item and its substitution is minimized.\;

Choose the remaining centers arbitrarily such that the fairness constraint is satisfied.\;

Return the $k$ centers.\;
\end{algorithm}

\section{$k$-centers for Range-Based Fairness}
\label{sec:offline}
To accomplish this, we need to elaborate on the fairness requirements on the set $g(A)$ in the adapted fair-shift problem. First, we remove the lower-bound fairness constraint on the set $g(A)$, since we could have $|A| < \sum_i l_i$. Second, we need to address the property in which the remaining $(k-|A|)$ centers can be filled arbitrarily. This guarantee is implicit in the equality-based case, since it follows from fairness as an upper bound. This is not generally true with range-based fairness case, since the set $g(A)$ may satisfy the upper bounds $u_i$ for all demographic groups but also have too few points in some demographic groups to satisfy all $l_i$ by assigning the remaining $k - |A|$ centers. With this note, the requirements on $g:A\rightarrow S$ for the range-based fair-shift problem are
\begin{itemize}
    \item $\forall a \in A, d(g(a),a) < d'$
    \item $\forall \text{ demographic groups $i$, }|S_i \cap g(A)| \le u_i$
    \item $\sum_i \max\{0,l_i - |S_i \cap g(A)|\} \le k - |A|$
\end{itemize}
where the first requirement limits the shift distance, the second enforces the $u_i$ constraints, and the third requirement allows us to enforce $l_i$ in the final $k$-centers by adding points.

\subsection{Testing for Fair-Shift}

\subsubsection{Designing the Algorithm}

We again use matching via maximum flow to solve the problem. Two significant changes are required to adapt to the new fairness constraint, specifically targeting the lower bound constraints. First, we add a new vertex that represents the $k - |A|$ centers which are not assigned by $g(A)$. This vertex connects to every demographic group node, since these arbitrary centers can be assigned to any group. For the second change, we enforce the bounds $(l_i, u_i)$ for each demographic group by adding flow constraints from the vertex representing group $i$ to the sink $t$. This initial construction $G_A$ can be seen in Figure \ref{fig:withlb}. We would compute maximum flow from $s$ to $t$, with a flow of $k$ representing a successful fair shift. The vertex set $V_A$ represents $A$, $V_f$ is the vertex set for demographic groups, and $v_c$ is responsible for the $k-|A|$ arbitrary centers.

\begin{figure}[ht]
    \centering
    \includegraphics[width=0.32\textwidth]{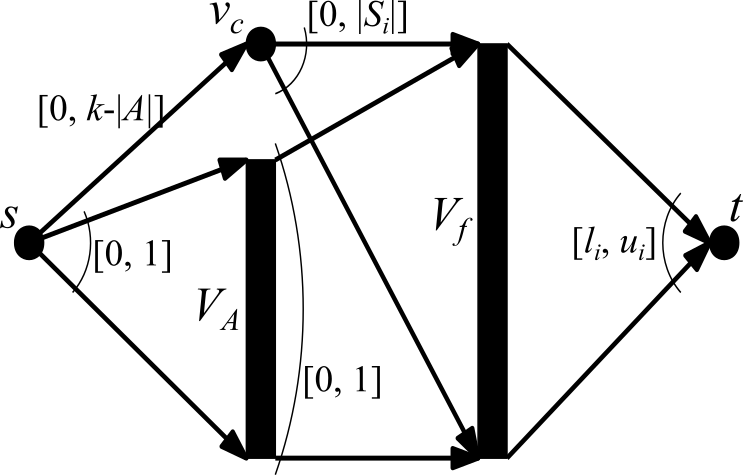}
    \caption{The original construction $G_A$. Edge capacities are given as [lower, upper]. Individual bounds with index $i$ correspond to the incident demographic-group node in $V_f$.}
    \label{fig:withlb}
\end{figure}

\begin{figure}[ht]
    \centering
    \includegraphics[width=0.37\textwidth]{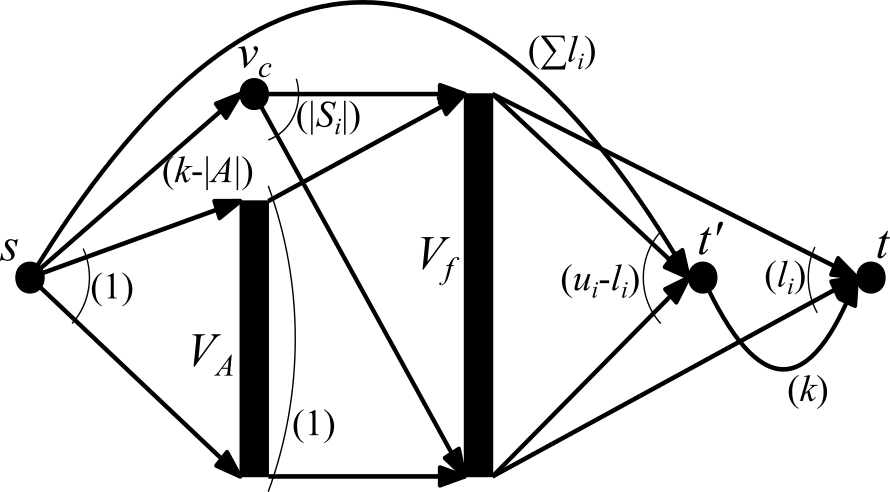}
    \caption{The construction $G'_A$ reducing $G_A$ to a problem without lower bounds. Upper-bound edge capacities are given as (capacity) to differentiate from vertex and vertex set labels. Individual bounds with index $i$ correspond to the incident demographic-group node in $V_f$.}
    \label{fig:withoutlb}
\end{figure}

%\begin{figure}[ht]
%  \def\svgwidth{\columnwidth}
%  \includesvg{extension_both}
%  \caption{The original construction $G$ with lower bounds (left) and the construction $G'$ without lower bounds (right). Edge capacities are given as [lower, upper]. Edges which are added in transition from $G$ to $G'$ are thicker and blue, edges with decreased capacity are dashed and red.}
%\label{fig:lowerboundconstruction}
%\end{figure}

% CHANGE
% In order to use Dinic's flow algorithm to solve this problem, we need to transform this into a flow problem without lower-bound constraints on edges. This transformation is from Section 7.7 of \cite{algorithmDesign}. First, add a new source and sink, which are incident to the old source and sink with edges whose capacity equals the previous minimum $s-t$ cut. Then, for each edge $e$ with a positive lower bound $l_i$, subtract $l_i$ from both the upper and lower bounds on $e$ and add edges from the new source to $e$'s destination and from $e$'s source to the new sink, each with lower bound capacity 0 and upper bound capacity $l_i$. Note that in our transformation we do not add a new source since this does not affect the flow value or the corresponding matching, and we also merge parallel edges and their capacities from $s$ to $t'$. The result of the transformation on our flow graph, which we call $G'_A$, is on the right of Figure \ref{fig:withoutlb}. We show this is a correct reduction:

In order to use Dinic's flow algorithm, we need to transform this into a flow problem without lower-bound constraints on edges. This transformation is from Section 7.7 of \cite{algorithmDesign}, and the resulting construction is in Figure \ref{fig:withoutlb}.

\begin{lemma}
    \label{lem:transformation_equality}
    The flow graph $G_A$ yields a valid $s$-$t$ flow of value $k$ iff the transformed flow graph $G'_A$ yields a valid $s$-$t$ flow of value $k + \sum_i l_i$.
\end{lemma}

The proof of this lemma follows from the fact that $(s, V\setminus s)$ is a minimum cut and the proofs in Section 7.7. of \cite{algorithmDesign}.

% CHANGE
%\begin{proof}
%    Suppose there exists a flow satisfying all constraints with flow value $k$ in $G_A$. For each edge without a lower bound in $G_A$, put the same flow through the same edge in $G'_A$. For each edge with a lower bound in $G_A$, the same edge in $G'_A$ gets the same flow minus the lower bound, and the edges added in the transformation for that edge get flow equal to the lower bound. Each node individually preserves conservation of flow, and each edge satisfies its upper bound in $G'_A$ either by the corresponding requirement in $G_A$, by seeing the same decrease in flow and upper-bound capacity from $G_A$ to $G'_A$, or by setting flow equal to capacity. The flow is increased by $l_i$ for each edge with a lower bound. Since the augmenting paths corresponding to these increases are edge-disjoint, the increases themselves are also disjoint, and therefore the constructed flow in $G'_A$ has value $k + \sum_i l_i$ and is a valid $s$-$t$ flow.
    
%    Instead suppose that the transformed graph $G'_A$ yields a flow of value $k + \sum_i l_i$ which satisfies all edge capacity constraints. Since the cuts $(s, V \setminus s)$ and $(t, V \setminus t)$ in $G'_A$ have value $k + \sum_i l_i$, these are minimum cuts, and all edges along these cuts are saturated in the flow. Therefore, since the edges added by the transformation are all saturated, we can reverse the flow transformation from $G_A$ to $G'_A$ given in the first half of the proof, and we obtain a valid flow with value $k$ in $G_A$ which satisfies all edge constraints, including lower bounds.
%\end{proof}

To solve the fair-shift problem, we construct the graph $G'_A$ and solve for a maximum $s$-$t$ flow. That is, the vertices $s, t, t', v_c$, the vertex sets $V_A$ and $V_f$ are added, and all edges not between $V_A$ and $V_f$ are added. For each vertex $a \in A$ and demographic group $i$, add an edge from the vertex for $a$ in $V_A$ to the vertex for $i$ in $V_f$ iff there exists an item $s \in S$ such that $s$ is in demographic group $i$ and $d(a,s) < d'$. 

There are three important notes about the construction. First, we will never actually build the graph $G_A$ in our $k$-centers algorithm but will instead directly build and test $G'_A$. Second, we build $G'_A$ from scratch in algorithm \ref{alg:newFairShiftTest}, but we will instead modify $G'_A$ incrementally during the binary searches in algorithm \ref{alg:3approx}, to save computation time. 
Third, because we only care about sufficiently small $d'$ as stated in section \ref{sec:prelim}, it suffices to compute the edges between $V_A$ and $V_f$ in $O(n)$ time by looking at each $s \in S$ and considering only its distance to the closest point in $A$.
We present algorithm \ref{alg:newFairShiftTest} to solve the fair-shift problem for range-based fairness.

\begin{algorithm}[t]
\caption{Testing the new fair shift constraint}\label{alg:newFairShiftTest}
\DontPrintSemicolon
\SetAlgoLined
\LinesNumbered
\KwIn{a set of points $S$, a set of points $A \subseteq S$, a value $d'$, a value $k$, teh bounds $l_i$ and $u_i$}
\KwOut{A fair set $B$ such that $|A| \ge |B|$ and $\forall a\in A, d(a,B) \le d'$, OR the empty set}
\BlankLine
Create a directed graph $G'_A = \{V,E\}$ with initialized values $V = \{s,v_C,t',t\}$, $E = \{(s,v_C,\text{cap.}=k-|A|), (t',t,\text{cap.}=k), (s,t',\text{cap.}=\sum_{i=1}^m l_i)\}$\;

Create a set of vertices $V_A$ with size $|A|$ which maps one-to-one with points in $A$\;

Create a set of vertices $V_f$ with size $m$ which maps one-to-one with demographic groups\;

Update $V$ with value $V \sqcup V_A \sqcup V_f$\;

\For{\text{all demographic group values} $f$}{
    Update $E$ with an edge from the corresponding vertex for $f$ in $V_f$ to $t$ with capacity $l_f$.\;
    
    Update $E$ with an edge from the corresponding vertex for $f$ in $V_f$ to $t'$ with capacity $u_f - l_f$.\;
    
    \For{$i = 1$ \KwTo $|A|$}{
        \If{$\min_{s \in S_i}d(s,a_i) < d'$}{
        Update $E$ with an edge from the corresponding value for $a_i$ in $V_A$ to the corresponding vertex for $f$ in $V_f$ with capacity 1 and label $s$
        }
    }
}

Find the maximum (integer) flow in $G'_A$ using Dinic's algorithm

\eIf{maximum flow value = $k + \sum_{i=1}^m l_i$}{
    Set $B = \emptyset$
    
    \For{$v_i \in V_A$ and $v_f \in V_F$ such that $(v_i, v_f) \in E$ has flow}{
        Update $B = B \cup \text{(the label of $(v_i, v_f)$)}$
    }
    return $B$
}{
    return $\emptyset$
}
\end{algorithm}

\subsubsection{Correctness and Runtime}

\begin{lemma}
    Suppose $d'$ is small enough so that balls of radius $d'$ around $A$ are disjoint.
    If there exists a function $g$ as specified for the range-based fair-shift problem on inputs $S$, $A$, $d'$, $k$, $\{l_i\}_{i\in[m]}$, $\{u_i\}_{i\in[m]}$ then algorithm \ref{alg:newFairShiftTest} returns the set $B = g(A)$ for some such function $g$. If no such function $g$ exists, then algorithm \ref{alg:newFairShiftTest} returns the empty set.
    \label{lem:fairshift_correctness}
\end{lemma}

\begin{proof}
    Note that this can be phrased as an if-and-only-if: the set $B = g(A)$ for some $g$ is returned if and only if some function $g$ exists. We prove it as such.
    
    Suppose algorithm 2 returns a set $B$. The matching for $B$ has one edge for each $a \in A$. That edge, and specifically its label, corresponds to some other point in $S$ such that $d(a,s) \le d'$. So, the points in $B$ matched to the incident vertex in $A$ yields the mapping $g$. The requirement that $|S_i \cap g(A)| \le u_i$ for all $i$, is met by constraints on edges from $V_f$ to $t$. The third requirement on $g$ is enforced by $v_c$: the flows from $v_c$ to $V_f$ yields an assignment for the arbitrary centers which will satisfy lower bounds. Therefore, a valid mapping $g$ exists, and the returned set is $B = g(A)$.
    
    Conversely, suppose the proper mapping $g$ exists. For each $a\in A$, the corresponding $g(a)$ satisfies $d(a,g(a)) \le d'$. Therefore, the edge from $a$'s representative in $V_A$ to the representative of $g(a)$'s demographic group in $V_f$ is present in $G'_A$, for all $a \in A$. Therefore, we obtain a flow of value $|A|$ in the corresponding $G_A$ by adding the flow through these edges along the $s-t$ paths following $s \rightarrow V_A \rightarrow V_f \rightarrow t' \rightarrow t$ for each $a \in A$. Additionally, by the second and third requirements of $g$, we can add flow along $s-t$ paths following $s \rightarrow v_c \rightarrow V_f \rightarrow t' \rightarrow t$ such that we satisfy all lower bounds, violate no upper bounds, and achieve total flow value $k$. By lemma \ref{lem:transformation_equality}, this means we can achieve the necessary flow in $G'_A$ to return a non-empty set $B$, and the set $B$ would be built by edges given by $g(A)$, since we designed a flow where all these edges have non-zero flow.
\end{proof}

    In runtime analysis, we ignore the time required to construct the graph $G'_A$, since this is handled specially in our $k$-centers algorithm. We specifically address the runtime of the fair-shift algorithm in the context of the $k$-centers algorithm. This implies that $|A| \le k$, $|V_f| = m \le k$, and there are at most $n$ edges from $V_A$ to $V_f$.

\begin{lemma}
    The time complexity of algorithm \ref{alg:newFairShiftTest} in the context of algorithm \ref{alg:3approx} is $O(n\sqrt{k} + T(n,k))$, where $T(n,k)$ is the time required to construct the graph $G'_A$.
    \label{lem:nsqrtk}
\end{lemma}

\begin{proof}
    To prove our claim, we need to show the time complexities of Dinic's algorithm on $G'_A$ and computing $B$ are $O(n\sqrt{k})$.
    To construct $B$, we look at all edges between $V_A$ and $V_f$. These edges are obtained from $S$, so there are at most $O(n)$ such edges. We do $O(1)$ work to check each $s\in S$, hence we construct $B$ is $O(n)$ time.
    
    % CHANGE
    %We analyze Dinic's algorithm on our graph as a special case, similar to the special case for matching. The noteworthy difference between our analysis and the special matching analysis of Dinic's is that augmenting paths in the same level-graph in this case do not have to be vertex- or edge-disjoint, since some edges have non-unit capacity in $G'_A$.
    
    Before we continue, we add some notation and asymptotic values. Let $E_{hi}$ be the set of edges in $G'_A$ with non-unit capacity, and let $E_{lo}$ be the set of unit-capacity edges. We have $O(m)$ vertices in $V_f$ and $O(k)$ vertices in $V_A$, as well as 4 other vertices, hence $|V| = O(k)$. In $E_{lo}$, we have $O(k)$ edges from $s$ to $V_A$ and $O(|S|) = O(n)$ edges from $V_A$ to $V_f$, so $|E_{lo}| = O(n)$. $E_{hi}$ has $O(m)$ edges from $G_A$ and an additional $O(m)$ edges from the construction of $G'_A$, hence $|E| = |E_{lo}| + |E_{hi}| = O(n)$.
    
    Recall that Dinic's algorithm proceeds in iterations: in each iteration, BFS is used to compute a layered graph and then DFS is used to find a maximal set of augmenting paths on the layered graph.
    The time to compute the layered graph per iteration is $O(|E|+|V|) = O(n)$, since it uses ordinary BFS. The DFS to compute the maximal set of shortest augmenting paths doesn't force vertex-disjoint maximal augmenting paths. In order to account for this, our DFS is be split into three main operations:
    \begin{itemize}
        \item Advance - Traverse any edge with remaining capacity.
        \item Retreat - When we cannot advance from a node other than $t$, move back over the last-traversed edge and delete it from this layer.
        \item Augment - When we reach $t$, we have an augmenting $s$-$t$ path. Traverse back along the path, decreasing the capacity of each edge by the minimum available capacity on the path, Delete edges from the layer which have 0 capacity.
    \end{itemize}
    The asymptotic total cost of the DFS is the number of times any of these operations are performed on a single edge, since a linear relationship exists between node operations and edge operations. Also, any time we Advance, we will always either Retreat back over the edge or that edge will be part of an Augment, hence the asymptotic time complexity of the DFS is bounded the number of times all edges are part of a Retreat or Augment operation.
    
    We have at most $|E| = O(n)$ Retreats or Augments on unit-capacity edges in any iteration. Notice that all of the edges in $E_{hi}$ are incident to at least one of 4 vertices: $s$, $v_C$, $t'$, and $t$. These vertices are in at most 4 layers of the layered graph, so we can only cross edges with non-unit capacities when we are entering or exiting these layers, hence $O(1)$-many times per augmenting path. Since each augmenting path increases the flow in $G'$ by at least 1, we can upper bound the total number of augmenting paths found during Dinic's on $G'_A$ by the maximum flow, which is $k + \sum_{i=1}^m l_i \le 2k = O(k)$. Therefore, the cost of Augments on high-capacity edges $O(1) \cdot O(k) = O(k)$ for all DFS's in the entirety of Dinic's algorithm, so the total time complexity of the BFS and DFS in a single iteration of the algorithm is $O(n) + O(k) = O(n)$.
    
    The remainder of the proof is very similar to the special case of Dinic's algorithm on bipartite-matching flow graphs. After $\sqrt{k}$ iterations in Dinic's algorithm, the (capacitated) symmetric difference between the current flow and the maximum flow can be decomposed into alternating cycles and paths which are vertex-disjoint except at $s$, $t$, and the endpoints of non-unit capacity edges. Since each augmenting path has length at least $\Omega(\sqrt{k})$ vertices at future iterations and uses at most $O(1)$ high-capacity edges, each augmenting path has $\Omega(\sqrt{k})$ unique vertices, and therefore there can be at most $O(k) / \Omega(\sqrt{k}) = O(\sqrt{k})$ remaining augmenting paths. Since each iteration increases the flow by at least 1, there can be only $O(\sqrt{k})$ remaining levels, and therefore $O(\sqrt{k})$ levels total.
    
    Therefore, we run BFS and DFS in $O(n)$ time $O(\sqrt{k})$-many times, for a total time complexity of $O(n\sqrt{k})$ to test for a fair shift.
\end{proof}

\subsection{Putting it Together}

% CHANGE
% Algorithm \ref{alg:3approx} is the outline for the fair $k$-centers problem, as discussed earlier. First, we apply the Gonzalez algorithm to select $k$ preliminary centers. Then, we binary search to find the maximum value $h \in \{1,...,k\}$ such that the first $h$ preliminary centers have a fair shift with low cost, and binary search to minimize the distance of the shift on those centers. Finally, we fill in the remaining centers arbitrarily to satisfy fairness. We first prove that the algorithm is a 3-approximation. As only the fair-shift problem and corresponding algorithm changed with respect to \cite{nguyen2020fair}, treating algorithm \ref{alg:newFairShiftTest} as a black-box with respect to the outline yields a proof which is incredibly similar to that of the algorithm in \cite{nguyen2020fair}. We use the notation $d_i$ and $a_i$ from Section  \ref{sec:gonz} to talk about the output of Step 1 in algorithm \ref{alg:3approx} here.

Algorithm \ref{alg:3approx} is the outline for the fair $k$-centers problem, as discussed earlier. In lines 2 and 3 we binary search to find the value $h$ and binary search to minimize the distance of the fair shift on the $h$ unfair centers. As only the fair-shift problem and corresponding test algorithm changed with respect to \cite{nguyen2020fair}, treating algorithm \ref{alg:newFairShiftTest} as a black-box with respect to the outline yields proofs which are incredibly similar to \cite{nguyen2020fair}. We use the notation $d_i$ and $a_i$ from Section \ref{sec:gonz} to talk about the output of Step 1 in algorithm \ref{alg:3approx} here.

%Therefore, as the number of sequence items we are considering increases, two things occur which limit the ability of the prefix with that size to pass the fair shift constraint. First, we add another item from the sequence which must be matched in order to shift centers and maintain fairness. Second, for every sequence item $a$ we already need to match, the number of points which are "near" $a$ is non-increasing, since $d_i / 2$ is non-increasing with $i$ by lemma \ref{seqdec}), so $a$ has no new points it can shift to, and possibly has less such points. In these ways, our matching becomes more restricted as the number of sequence items we are considering increases.

%Let the optimal solution of a specific $k$-centers problem with fairness require balls of radius $r^*$ around the centers to cover $S$. On the same problem, we give the following lemma and prove it directly:

First, the following lemma allows us to bound $d_{h+1}$:

\begin{lemma}\cite[Lemma 5.2]{nguyen2020fair}
$d_i > 2r^* \implies $ the set $\{a_j\}_{j \le i}$ satisfies the range-based fair-shift constraint with $d'=d_i/2$\label{fairshift}.
\end{lemma}

From there, we can show that the set found by shifting $\{a_j\}_{j \le h}$ gives us a 3-approximation:

\begin{lemma}\cite[Lemma 5.3]{nguyen2020fair}
Let $C_h$ be the set of points returned by Step 3 of algorithm \ref{alg:3approx}. Then, any set $C$ such that
\begin{itemize}
    \item $C$ satisfies the fairness constraint
    \item $|C| = k$
    \item $C_h \subseteq C$
\end{itemize}
is a 3-approximation for the $k$-centers problem with range-based fairness.\label{sub3approx}
\end{lemma}

Since the arbitrary centers in Step 4 of algorithm \ref{alg:3approx} with the set $g(A)$ yield $k$ centers which satisfy fairness, lemma \ref{sub3approx} immediately yields the following theorem:
\begin{theorem}
\label{thm:3approx}
Algorithm \ref{alg:3approx} gives a 3-approximation for the problem of $k$-centers with fairness.\label{3approx}\\
\end{theorem}

% CHANGE
%We do not include the proofs of lemmas \ref{fairshift} and \ref{sub3approx} because the only requirement not immediately implied by the similar structure of the algorithms is that the set $g(A)$ can be expanded arbitrarily to a fair set of $k$ centers. Since we ensured this with the third condition in the fair-shift problem, the proofs of these lemmas are identical to their proofs in \cite{nguyen2020fair}.

The proofs of lemmas \ref{fairshift} and \ref{sub3approx} follow immediately from lemmas in \cite{nguyen2020fair} because the property that $g(A)$ can be expanded arbitrarily to a fair set of $k$ centers is maintained.

%Finally, we justify the runtime of our algorithm. Again, we will take advantage of similarities with the equality-based fair $k$-centers in \cite{nguyen2020fair}. Step 1 is identical to the equality-based case, and Step 4 is $O(nk)$ even in the brute-force approach of sweeping $S$ $O(1)$-many times for each demographic group. As in \cite{nguyen2020fair}, we will apply a binary search at both Steps 2 and 3. The range of the binary searches, the computation to find the range in the case of the second binary search, and the decision rules with respect to the valid search ranges in both binary searches are identical. The process for graph constructions, which are accomplished by maintaining, copying, and adding to and removing from a base-level graph, are almost identical in the two cases. The only difference with respect to graph construction is that the new edge $(s,v_c)$ will need to update its capacity every time $|A|$ changes between levels in the first binary search. This change is $O(1)$ on at most $O(\log k)$ levels, and therefore has an asymptotically insignificant contribution of $O(\log k)$ overall. Therefore, we can take advantage of the proof of and discussion on theorem 5.5 in \cite{nguyen2020fair}, along with lemma \ref{lem:nsqrtk} as a proof for the runtime of Dinic's algorithm, to obtain our total runtime:

Again, we take advantage of similarities with \cite{nguyen2020fair} for the runtime proof. Step 4 is $O(nk)$ even in the brute-force approach of sweeping $S$ $O(1)$-many times for each demographic group. The only other difference between algorithm \ref{alg:3approx} and \cite{nguyen2020fair} is during graph construction, where the weight of the edge $(s,v_c)$ must be updated at every level of the binary search. This takes $O(1)$ time on $O(\log k)$ levels, and therefore has an asymptotically insignificant contribution of $O(\log k)$ overall. Thus, we can take advantage of the proof of and discussion on theorem 5.5 in \cite{nguyen2020fair}, along with lemma \ref{lem:nsqrtk}, to obtain our total runtime:

\begin{theorem}
    \label{thm:fair_nktime}
    Algorithm \ref{alg:3approx} runs in $O(nk)$ time.\\
\end{theorem}

We combine this with theorem \ref{thm:3approx} to obtain the main result of the section:

\begin{theorem}
    \label{thm:gen_kcenters_result}
    $k$-centers with range-based fairness has a 3-approximation algorithm with running time $O(nk)$ given by algorithm \ref{alg:3approx}.
\end{theorem}

\section{Streaming Fair $k$-centers}
In this section, we combine our classical algorithms in the previous section with techniques of \cite{guha2009tight, kale2019small} to give a streaming algorithm for range-based fair $k$ centers. The main building block of the algorithm is a subroutine that is given a guess $R$ of the optimal radius $\OPT$ and processes the stream based on this guess. This subroutine (algorithm \ref{alg:instance}) parameterized by the guess $R$ is denoted by $\mathbb{I}(R)$.

We run several instances of the subroutine with different guesses in parallel. At all times, we keep track of a lower bound $\tau$ for $\OPT$. Initially $\tau$ is half of the minimum distance between two points in the first $k+1$ points in the stream. Let $\tau_{\min}$ be the minimum power of $(1+\eps)$ that is at least $\tau$.
The set of instances we run have guess values $\tau_{\min}, \tau_{\min} (1+\eps), \tau_{\min}(1+\eps)^2,..., \tau_{\min}(2+\eps)/\eps=\tau_{\max}$. As the lower bound is updated during the stream, subroutines with guesses that are too small, are aborted and new subroutines with higher guesses are started. More precisely, if a new lower bound $\tau'>\tau_{\min}$ is found, we abort all instances $\tau_{\min}(1+\eps)^i < \tau'$ and use them to create new instances with guesses between $\max(\tau', \tau_{\max}(1+\eps))$ and $(2+\eps)\tau'/\eps$. At the end of the stream, if one of the guesses is in $[\OPT, (1+\eps) \OPT]$, the algorithm can find a fair solution with cost at most $(1+\eps) (13+O(\eps)) \OPT$. Otherwise, our lower bound must be much smaller than \OPT. In this case, we can run our offline algorithm on a stored subset of the input and find a solution with even better approximation factor.  

Next we consider the implementation of the subroutine $\mathbb{I}(\cdot)$. Each instance $\mathbb{I}(\Delta)$ of algorithm \ref{alg:instance} maintains a set $C(\Delta)$ of centers, so-called ``pivots'', and each pivot $c$ has a disjoint set $I_c$ with up to $1$ point from demographic group $i$. The set $I_c$ is intended to be replacement candidates for $c$ to fulfill the fairness constraint. All points in $I_c$ are guaranteed to be within distance $(2+\eps)\Delta$ from $c$ . Let $I(\Delta)=\{(c,I_c)\ |\ c\in C(\Delta)\}$. Note that $I(\Delta)$ can be viewed equivalently as a collection of $(key, value)$ pairs or a mapping $key\to value$ i.e. $I(\Delta)(c) = I_c~\forall c\in C(\Delta)$. It also stores a set $C_n(\Delta)$ of up to $u_i$ points from $S_i~\forall i$ to complete the final solution for the fairness lower bounds. 

After every $k$ insertions, the algorithm checks if there is an instance with more than $k$ centers. If so, the algorithm computes a new lower bound for $\OPT$, aborts and replaces instances with respect to the new lower bound. When an instance $\mathbb{I}(\Delta)$ is aborted and replaced with a new instance with a larger guess, the new instance inherits a set of pivots from the previous instance, $C(\Delta)$, and the associated collection of sets, $I(\Delta)$. The algorithm processes the set of pivots so that there are at most $k$ points in the new set of pivots. The number of pivots increases by at most $1$ per insertion and the algorithm checks every $k$ insertions so the number of pivots is at most $2k$ at all times.

The subroutine $\mathbb{I}(\Delta)$ processes a new point $p$ in the stream as follows. If $p$ is within distance $2\Delta$ from a pivot $c$, then it is added to the set $I_c$ if $I_c$ does not have a point in the same demographic group. If $p$ is farther than $2\Delta$ from all pivots, then it is added as a new pivot. $p$ is also added to $C_n(\Delta)$ if $p\in S_i$ and $|C_n(\Delta)\cap S_i|< u_i$. 

If the subroutine $\mathbb{I}(\Delta)$ is not aborted mid-stream then we run algorithm \ref{alg:merge} on $C(\Delta)$ and $I(\Delta)$ to compute a set of centers $C_m$ that satisfy the fair-shift constraint and use $C_n$ to add centers to $C_m$ to satisfy the fairness lower bounds. If this step fails for every instance, then we can run an offline algorithm on the stored points $I(\tau_{\min})\cup C(\tau_{\min})$ from $\mathbb{I}(\tau_{\min})$.

Algorithm \ref{alg:merge} first constructs a subset $S$ of $C$ such that 1) any point in $C$ is within $(6+2\eps) \Delta$ of $S$ and 2) all points in $S$ are at least $(6+2\eps)\Delta$ apart. Next, for each point $s$ in $S$, we create a new set $I'_s$ that is the union of $I_c$ for all $c\in C$ within distance $(3+\eps)\Delta$ from $s$. Notice that $I'_s$ are disjoint. The key observation here is that if $\Delta \in [\OPT, (1+\eps) \OPT]$ then when we move all points in $I(\Delta)(c)$ to $c$, there is a fair-shift solution for the centers $S$ using replacements sets $I'_s$. A fair solution can be obtained by augmenting the results with centers from $C_n$ to fulfill the fairness constraint without hurting the solution quality. Each center in $S$ could be at distance as far as $(3+\eps)\Delta + (2+\eps)\Delta = (5+2\eps)\Delta$ from its replacement. Furthermore, centers in $C$ could be at distance $(6+2\eps)\Delta$ from $S$ so overall, the final solution covers $C$ at distance $(11+4\eps)\Delta$. Each point in $C$ represents a cluster of radius $(2+\eps)\Delta$ so the final solution covers all points at distance $(13+5\eps)\Delta$.

If none of our guesses is in the range $[\OPT, (1+\eps) \OPT]$ then it must be the case that $\OPT \gg \tau_{\min}/\eps$ and we can simply run our offline algorithm on the stored points in $I(\tau_{\min})$. Because all points are within distance $O(\tau_{\min})$ from the stored points and the offline algorithm gives a $3$-approximation, the result is a $3+O(\eps)$-approximation.

\begin{algorithm}[htb!]
\caption{1-pass fair $k$-centers}
\label{alg:1pass}
\SetKwFunction{Process}{Process}
\SetKwFunction{MergeCenters}{MergeCenters}
\SetKwFunction{Replacement}{Replacement}
\SetKwProg{Fn}{Function}{:}{}
\DontPrintSemicolon
\SetAlgoLined
Let $\tau$ be half the minimum distance among the first $k+1$ points in the stream.\;
Let $\tau_{\min}$ be the minimum power of $(1+\eps)$ that is at least $\tau$\;
$G = \{\tau_{\min}, \tau_{\min}(1+\eps),\ldots, \tau_{\min}(1+\eps)^\beta = (2+\eps)\tau_{\min}/\eps\}$\;
$\tau_{\max} = \tau_{\min}(1+\eps)^\beta$\;
\For{$\Delta \in G$} {
  Initialize $\mathbb{I} (\Delta)$ with $C(\Delta)\gets \emptyset, I(\Delta)\gets \emptyset, C_n(\Delta)\gets \emptyset$\;
  \Process($\Delta,p, \{p\}$) $\forall p$ in the first $k+1$ points\;
}
\For{points $p_i$ after first $k+1$ points} {  
  \For{$\Delta \in G$}{
    \Process($\Delta, p_i, \{p_i\}$)\;
  }
  \If{($i \equiv 0 \pmod{k}$ or $p_i$ is the final point) and ($\exists \Delta\in G$ s.t. $|C(\Delta)|> k$)}{
    $\tau_r\gets \tau$\;
    \For{$\Delta\in G$}{
      \If{$|C(\Delta)|>k$}{
        Run Gonzalez on $C(\Delta)$ and compute $\tau' = d_{k+1}/2$\;
        $\tau_r\gets \max(\tau_r, \tau')$\;
      }
    }
    $\tau\gets \tau_r$\;
    $\tau_{old}\gets \tau_{\min}$\;
    Update $\tau_{\min}$ to the minimum power of $(1+\eps)$ that is at least $\tau_r$\;
    $\tau_{\max}\gets \tau_{\min}(1+\eps)^{\beta}$\;
    \For{$\Delta \in \{\tau_{\min}, \tau_{\min}(1+\eps), \ldots, \tau_{\max}\}\setminus G$} {\label{line:begin-new-stream}
      Initialize a new instance $\mathbb{I}(\Delta)$ with $C_n(\Delta)\gets C_n(\tau_{old})$\;
      Run \Process($\Delta, p, I(\tau_{old})(p)$) for $p \in C(\tau_{old})$ in the order they are selected by the Gonzalez algorithm on input $C(\tau_{old})$ (first $k+1$ points in the order they are selected then the rest arbitrarily)\;\label{line:end-new-stream}
    }
    Abort all $\mathbb{I}(\tau)\in G$ with $\tau < \tau_{\min}$\;
    $G\gets \{\tau_{\min}, \tau_{\min}(1+\eps), \ldots, \tau_{\max}\}$\;
  }
}
\For{$\Delta\in G$}{
  $C_m =  \MergeCenters(\Delta, C(\Delta), I(\Delta))$\;
  Add centers from $C_n$ to $C_m$ so $C_m$ satisfies fairness's lower bounds if $|C_m| \neq \emptyset$ otherwise \KwRet ``\textbf{failure}''\;
  \KwRet $C_m$\;
}
Run the offline algorithm on $I(\tau_{\min})\cup C(\tau_{\min}) \cup C_n(\tau_{\min})$ \label{line:offlinecall}\;
\end{algorithm}

\begin{algorithm}[htb!]
\caption{Process($\Delta, p, I$): Subroutine $\mathbb{I}(\Delta)$ processes a point $p$ and its associated set $I$.}
\label{alg:instance}
\SetKwFunction{StreamKcenters}{StreamKcenters}
\SetKwFunction{MergeCenters}{MergeCenters}
\SetKwFunction{Replacement}{Replacement}
\SetKwProg{Fn}{Function}{:}{}
\DontPrintSemicolon
\SetAlgoLined

\uIf{$\exists c\in C(\Delta)$ s.t. $d(c, p) \le 2\Delta$}{
  \tcp{Add $p$ to $c$'s cluster}
  $I(\Delta)(c)\gets I(\Delta)(c)\cup I$\;
  Remove points from $I(\Delta)(c)$ if needed so that $|I(\Delta)(c)\cap S_i|\le 1~\forall i$\;
}
\Else(\tcp*[h]{Add new pivot}){
  $C(\Delta) \gets C(\Delta) \cup \{p\}$ \;
  $I(\Delta) \gets I(\Delta) \cup \{(p, I)\}$ \;
}
Add $I$ to $C_n(\Delta)$ and trim it so that $|C_n(\Delta)\cap S_i|\le u_i~\forall i$\;
\end{algorithm}
\begin{algorithm}[htb!]
\caption{MergeCenters$(\Delta)$}
\label{alg:merge}
\SetKwFunction{StreamKcenters}{StreamKcenters}
\SetKwFunction{MergeCenters}{MergeCenters}
\SetKwFunction{neighbor}{neighbor}
\SetKwProg{Fn}{Function}{:}{}
\DontPrintSemicolon
\SetAlgoLined
$D \leftarrow \infty$\;
Let $c$ be an arbitrary point in $C(\Delta)$\;
$C \leftarrow \{c\}$\;
\While{$D > (6+2\eps) \Delta$ and $|C| < |C(\Delta)|$}{
  $c = \argmax_{ c \in C(\Delta)}{ \min_{c' \in C} d(c, c')}$\;
  $D \leftarrow  \min_{c' \in C} d(c, c') $\;
  \If{$D > (6+2\eps) \Delta$}{
    $C \leftarrow C \cup \{c\}$\;
  }
}
$D \gets \min_{c\ne c' \in C} d(c,c')$\;
\For{$c \in C$}{
  $I'_c \leftarrow  \cup \{I(\Delta)(c') : d(c,c') \leq (3+\eps) \Delta\}$\;
}
Keep in $I'_c$ at most one element from each group $S_i$ and remove the rest. \label{line:afterfilter}\;
Run algorithm \ref{alg:newFairShiftTest} with $S = \cup_{c \in C} I'_c, A = C_0, d' = D/2$ and pretending that all points in $I(\Delta)(c')$ are located at $c'$\;
\end{algorithm}

\subsection{Analysis}
We first analyze the running time and space of algorithm \ref{alg:1pass}. We start by showing that every time we update our lower bound estimate from $\tau$ to $\tau_r$, it must be the case that $\tau \le \tau_r \le \OPT$. 

\begin{lemma}
\label{lem:2opt}
$\tau_r \leq \OPT$
\end{lemma}

\begin{proof}
Recall that $\tau_r$ is half the distance between the $k+1$st point selected by the Gonzalez algorithm and the previous $k$ points. By lemma \ref{seqdec}, this is also half the minimum distance among the first $k+1$ points selected by the Gonzalez algorithm. Since we have $k+1$ points, by the pigeonhole principle, there exists two points $c_i,c_j$ that belong to the same cluster in the optimal solution. This implies $d(c_i, c_j) \leq 2 \OPT$. By lemma \ref{seqdec}, we have $\tau_r = d_{k+1} /2 \leq d(c_i, c_j)/2$. 
\end{proof}

\begin{lemma}
\label{lem:increase}
$\tau \le \tau_r$
\end{lemma}

\begin{proof}
By construction, all points in $C(\Delta)$ are at distance at least $2\Delta$ apart. Therefore, when running the Gonzalez algorithm on $C(\Delta)$, the distance computed $d_{k+1}=2\tau_r$ is at least $2\Delta$, which in turn is at least $2\tau$.
\end{proof}

To bound the running time, we show that there are at most $k$ pivots in each new instance after processing the pivots of the old instance.  
\begin{lemma}
\label{lem:mostkpoints}
Each new instance $\mathbb{I}(\Delta)$ created between lines \ref{line:begin-new-stream} and \ref{line:end-new-stream}
has at most $k$ pivots after processing the pivots from the old instance.
\end{lemma}

\begin{proof}
Since the points are processed in the order they are selected by the Gonzalez algorithm and by the property of Lemma~\ref{seqdec}, a prefix of the sequence of pivots are selected as new pivots and it stops at some index $i$ when $d_i \le 2\Delta$ (since $d_i$ is a decreasing sequence and all points from index $i$ onward are within distance $d_i$ from previous points). By the construction of $\tau_r$, we have $2\Delta \ge 2\tau_r\ge d_{k+1}$. Thus, all old pivots from the $k+1$st onward are not selected as new pivots.
\end{proof}

\begin{lemma}
\label{lem:memoryruntime}
Algorithm \ref{alg:1pass} stores $O((k m+\sum_i u_i)\log(1/\eps)/\eps)$ points and has running time $O\left(\left(\log(1/\eps) /\eps \right)  nk\right)$.
\end{lemma}

\begin{proof}
The number of pivots in each instance is reduced to at most $k$ after every $k$ insertions and it increases by at most $1$ after every insertion so the number of pivots is at most $2k$ at all times. Each pivot has a associated set of at most $ m$ points since it stores at most $1$ point from $S_i$ for each demographic group. We also store the set $C_n$, which has size at most $\sum_{i =1}^m u_i $. Since we have $O(\log (1/\eps)/\eps)$ instances,  the algorithm stores  $O\left(\left(km+\sum_{i =1}^m u_i\right)\log (1/\eps)/\eps\right)$ points.

 For every point in the stream, we run algorithm \ref{alg:instance} with runtime $O(k)$ for each instance. This takes $O(nk \log (1/\eps)/\eps)$  time overall. Note that for every $k$ points in the stream, we run Gonzalez on the set of pivot for of instance that has more than $k$ pivots. This takes $\left( k^2 \log (1/\eps)/\eps\right)$ since there are at most $2k$ pivots in each instance. Then, we run  a sub-routine (algorithm \ref{alg:instance})  for each instance  to process the pivots from $\mathbb{I}(\tau_{old})$. Since $|C(\tau_{old})| \leq 2k$ , each instance takes $O(k^2)$ time for every k points. Thus, for any set of $k$ contiguous points, the amortized time for running Gonzalez and the sub-routine to process the pivots from $\mathbb{I}(\tau_{old})$  is $\frac{O(k^2  \log (1/\eps)/\eps))}{k} = O(k \log (1/\eps)/\eps)$.      For algorithm \ref{alg:merge}, running the max flow program each time in algorithm $\ref{alg:newFairShiftTest}$ takes $O(n \sqrt{k})$ by lemma \ref{lem:nsqrtk} and the fact that the algorithm only takes a subset of the whole stream as input. Thus, it takes $O(n \sqrt{k} \log(1/\eps) / \eps)$ overall.   If at the end of algorithm \ref{alg:1pass}, all the instances report failure, then we run an offline algorithm from section \ref{sec:offline} on the subset of the stream. This takes $O(nk)$. Thus, the runtime for algorithm \ref{alg:1pass} is  $O(n k \log(1/\eps) / \eps)$.
\end{proof}

To bound the approximation factor, we first introduce some notation. For an instance $\mathbb{I}(\Delta)$, let $E(\Delta)$ be the set of points processed by the instance from its creation (beginning of the stream or when the lower bound is updated) to its termination (end of the stream or when the lower bound is updated). Let $C_o(\Delta)$ denote the set of pivots of instance $\mathbb{I}(\tau_{old})$ that is used to initialize $\mathbb{I}(\Delta)$ from line  \ref{line:begin-new-stream} to line \ref{line:end-new-stream}. If $\mathbb{I}(\Delta)$ starts from the beginning of the stream then $C_o(\Delta)=\emptyset$. 

\begin{lemma}
\label{lemma:approxloss}
There exists a mapping $f$ from points that arrived before $E(\Delta)$ to points in $C_o(\Delta)$ with the following properties.
First, for any point $e$ arriving before $E(\Delta)$, we have $d(e,f(e)) \leq \eps \Delta$. Second, for any pivot $c\in C_o(\Delta)$, we have $|I(\tau_{old})(c)\cap S_i| \ge \min(|f^{-1}(c) \cap S_i|, 1)$. 
\end{lemma}

\begin{proof}
  We prove this by induction on the number of times an instance $I(\Delta)$ is re-spawned (or aborted) with a different radius guess after $e$ is processed. When a point $p$ is inserted, the value $f(p)$ is initialized to its pivot and when a new instance $\mathbb{I}(\Delta)$ is created from pivots in $\mathbb{I}(\tau_{old})$, if an old pivot $p$ is added to the cluster of a new pivot $p'$, then all points $e$ with $f(e)=p$ are updated so that $f(e)=p'$.
  
  If $\mathbb{I}(\Delta)$ starts from the beginning of the stream then there is no point arriving before $E(\Delta)$ and the lemma vacuously holds. Next, consider the case $\mathbb{I}(\Delta)$ is spawned using the pivots from $\mathbb{I}(\tau_{old})$. 
  %and $\tau$ is previously used to initiate instance $V(\tau)$ (also referred as $\mathbb{I}(\tau')$),   
  Then there are two cases: either $e\in E(\tau_{old})$ or $e$ arrived
  before $E(\tau_{old})$.  If $e\in E(\tau_{old})$, there exists a pivot in the old instance $a(e) \in C(\tau_{old})=C_0(\Delta)$ such that:
  \begin{align*}
        d(e,a(e))\le 2\guess_{old} \le 2 \eps\Delta/(2+\eps) \le \eps \Delta
  \end{align*}
  
  The second inequality holds because we spawn a new instance with guess $\Delta$ only when $\Delta \geq (2+\eps) \tau_{old} /\eps$. 
  
  Otherwise, by the inductive hypothesis there is a point
  $e'\in C_o(\tau_{old})$ such that $d(e,e') \le \eps\tau_{old}$. By the logic of Algorithm~\ref{alg:instance}, there exists a pivot $c \in C(\tau_{old})=C_o(\Delta)$ such that $d(c,e') \leq 2 \tau_{old}$. By the triangle inequality,
  \begin{align*}
      d(e,c) \leq d(e,e') + d(c,e') \leq \eps \tau_{old} + 2 \tau_{old}& \\\le (2+\eps) \frac{\eps \Delta}{2+\eps} = \eps \Delta.
  \end{align*}
  
  For the second part of the lemma, observe that the associated sets $I(\Delta)$ are updated exactly the same way as we update the mapping function $f$. The only exception is that each set drops points in demographic $i$ when it already has $1$ point from $S_i$. Thus, the number of points in demographic $i$ that are retained in $I(\tau_{old})(c)$ is $\min(1, |f^{-1}(c)\cap S_i|)$.
  
\end{proof}

From lemma \ref{lemma:approxloss}, we can apply triangle inequality on $d(e,p)$ and the distance between $p$ and its pivot in $C(\Delta)$  to get the following corollary. 
\begin{corollary}
\label{lemma:coroll}
There is a mapping $f$ from points arriving before and during the execution of $\mathbb{I}(\Delta)$ to $C(\Delta)$ with the following properties. If $e$ is a point that arrives before or during the execution of $\mathbb{I}(\Delta)$, then $d(e, f(e))\le (2+\eps)\Delta$. Furthermore, for any pivot $c\in C(\Delta)$, we have $|I(\Delta)(c)\cap S_i|\ge \min(|f^{-1}(c) \cap S_i|, 1)~\forall i$.
\end{corollary}
%\begin{corollary}
%For any point $e$ in the stream that does not arrive after $E(\Delta)$, 
%\end{corollary}

\begin{lemma}
\label{lem:worstcase}
Suppose there is an instance $\mathbb{I}(\Delta)$ with $\Delta \in [\OPT, (1+\eps) \OPT]$. At the end of the stream for this instance, MergeCenters (algorithm \ref{alg:merge}) finds a set of centers $C_m$ satisfying the fair-shift condition and $d(C_m, e) \leq (13 + O(\eps)) \OPT~\forall e\in S$. As a consequence, the final output of the algorithm \ref{alg:1pass} satisfies fairness.
\end{lemma}

\begin{proof}
Note that the set of pivots $C_o = C(\Delta)$ is passed to algorithm \ref{alg:merge} at the end of the stream. %We first show for all points $c_o \in C_o$,  $d(C_m, c_o) \leq 11(1+\eps)\Delta$. After line \ref{line:afterfilter}, we show that there exists a set $A$ formed by selecting at most one point from each set $I'_c$ such that we can add points from $C_n$ to $A$ so that the set satisfies fairness, and we can use the algorithm \ref{alg:newFairShiftTest} to find such a set.

By corollary \ref{lemma:coroll}, there is a mapping $f:S\to C(\Delta)=C_o$ such that for all $o^*\in \OPT$, we have $d(o^*, f(o^*)) \leq (2+\eps) \Delta$. Furthermore, for every pivot $c\in C_o$, we have $|I(\Delta)(c)\cap S_i|\ge \min(|f^{-1}(c) \cap S_i|, 1)~\forall i$. Therefore, there is a mapping $g:\OPT\to \cup_c I(\Delta)(c)$ such that for every $o^*\in \OPT$, we have $g(o^*)\in I(\Delta)(f(o^*))$ and $o^*$ and $g(o^*)$ are in the same demographic group. Note that multiple points in $\OPT$ could be mapped by $g$ to the same point. By triangle inequality, $d(o^*, g(o^*)) \le d(o^*, f(o^*))+d(f(o^*), g(o^*)) \le (4+2\eps)\Delta$.

For every pivot $c_o \in C_o$, there is a center $c_o^* \in \OPT$ such that $d(c_o, c^*_o) \leq \OPT \leq \Delta$. There are two possibilities for $f(c_o^*)$: either $f(c_o^*)=c_o$ or $f(c_o^*)=c_o'$ where $d(c_o, c_o') \le d(c_o, c_o^*) + d(c_o^*, f(c_o^*)) \le (3+\eps)\Delta$.
%If $I_{c_o}$ covers $c_o^*$, then we're done. Otherwise, there are two cases when $I_{c_o}$ does not cover $c_o^*$. The first case is when $I_{c_o}$ already has enough of $c_o^*$'s demographic group and does not add it into the set ($(f(c^*_o) = c_o$). In this case, we can get a replacement $c_o^r$ for $c_o^*$ in $I_{c_o}$ such that they both share the same demographic by lemma \ref{lemma:approxloss}. In the second case, when $f(c^*_o) \neq c_o$  and $(f(c^*_o)$ is set to another pivot $c_o'$ instead, we know $c_o'$ has its associated set $I_{c_o'}$ takes either $c^*_o$ or its replacement $c_o^r$. Now, our job is to merge the sets $I_{c_o}$ and $I_{c_o'}$ into a set $I'_{c_o}$ so that $c_o$ can be swapped to $c^*_o$ or its replacement $c^r_o$. We will show that this is done by the construction of the set $C$ in algorithm \ref{alg:merge}.  From this point forward, we use $c^r_o$ interchangeably with $c^*_o$ in the analysis. 

By the logic of the algorithm, we know that every point in $C$ is at distance more than $(6+2\eps)\Delta$ from the others. For each new center $c\in C$, its new associated set $I'_c$ is the union of all $I(\Delta)(c')$ such that $d(c,c') \le (3+\eps)\Delta$. Observe that 1) each set in $I(\Delta)$ goes to at most one new set $I'_c$ and 2) if $c^*$ is the closest center in $\OPT$ to $c$ then $g(c^*)\in I'_c$.

By the above observation, there exists a fair-shift solution for the set of centers $C$ using disjoint replacement candidate sets $I'_c$ (replacing each $c$ with $g(c^*)$) and Algorithm~\ref{alg:newFairShiftTest} can find a solution. Let $h$ be the function that maps from $c\in C$ to its replacement $h(c)\in C_m$. The final solution is formed by taking the output $C_m$ of Algorithm~\ref{alg:newFairShiftTest} and adding arbitrary centers from $C_n$ to fulfill the fairness constraint. 

Next we bound the approximation factor. By construction, for each $c_o \in C_o$ there exists $c\in C$ such that $d(c_o, c) \le (6+2\eps)\Delta$. We have
\begin{align*}
    d(c_o, C_m) \le d(c_o, h(c)) \le d(c_o, c) + d(c, h(c)) \\
    \le (6+2\eps)\Delta + (3+\eps)\Delta + (2+\eps)\Delta= (11+4\eps)\Delta
\end{align*} 
By corollary \ref{lemma:coroll}, for any $p\in S$, we have $d(p, f(p))\le (2+\eps)\Delta$. By triangle inequality and the fact that $\Delta \leq (1+\eps)\OPT$, we have:
\begin{align*}
    d(p, C_m) & \leq d(p,f(p)) + d(f(p), C_m)  \leq (13+O(\eps)) \OPT.
\end{align*}
\end{proof}

\begin{theorem}
\label{theorem:streaming}
There is a $(13+O(\eps))$-approximation one-pass streaming algorithm for range-based fair $k$-centers that stores at most $O\left(km  \log(1/\eps) /\eps \right)$ points with runtime $O(nk\log(1/\eps)/\eps)$. 
\end{theorem}

\begin{proof}
Note that by lemma \ref{lem:2opt} and the fact that $\tau_{\min} \leq (1+\eps) \tau_r$, we have $\tau_{\min} \leq (1+\eps) \OPT$. If  there exist $p \in [\beta]$ such that $\tau_{\min} (1+\eps)^p \geq \OPT$, then the claim follows by lemmas \ref{lem:worstcase} and \ref{lem:memoryruntime}. 

If all the guesses return failure, it must be the case that $\tau_{\max}=(2+\eps)\tau_{\min}/\eps \le OPT$. At line \ref{line:offlinecall} of algorithm \ref{alg:1pass}, we run an offline fair $k$-center on the set $I=\cup_{c\in C(\tau_{\min})}I(\tau_{\min})(c) \cup C_n(\tau_{\min}) \cup C(\tau_{\min})$. The running time is $O(nk)$ since $I$ only contains a subset of all the points in the stream. We will show that the optimal fair $k$-center cost for input $I$ is $\OPT + \tau_{\min} (4+2\eps)$. This is because we can first select the optimal solution in the original problem and move its centers to points in $I$ in the same demographic groups as follows. Consider a point $c^*\in \OPT$. By Corollary \ref{lemma:coroll}, there exists some center $c \in C(\tau_{\min})$ such that $d(c^*,c) \leq (2+\eps)\tau_{\min}$ and furthermore, there exists $c'=c'(c^*)\in I(\tau_{\min})(c)$ such that $d(c', c) \leq (2+\eps) \tau_{\min}$ and $c'$ and $c^*$ are from the same demographic group. By triangle inequality, $d(c^*, c') \le  (4+2\eps)\tau_{\min}$.
The solution $\{c'(c^*)\ |\ c^*\in \OPT\}$ could have size smaller than $k$ when multiple points in $\OPT$ map to the same point in $I$ but it can be extended to a fair solution $S$ by adding points from $C_n(\tau_{\min})$ appropriately to satisfy the fairness lower bounds.

Because the offline algorithm gives a $3$-approxmation solution by theorem \ref{thm:gen_kcenters_result}, we have a fair solution $S$ with cost at most $(3 + 3\epsilon \frac{4+2\epsilon}{2+\epsilon}) \OPT = (3+6\eps)\OPT$. By Corollary~\ref{lemma:coroll}, for any point $p\in S$, there exists $c\in C(\tau_{\min})$ such that $d(c,p) \leq (2+\eps) \tau_{\min}$. Therefore,
\begin{align*}
    d(p,S) &\le d(c,S) +  d(p,c) \le (3+6\eps) \OPT + (2+\eps) \tau_{\min} \\
    &\le (3+7\eps)\OPT 
\end{align*}
\end{proof}

When the centers must contain an exact number from each demographic group, we have a special case of the range-based fairness $k$-centers problem where $u_i = l_i$ and $\sum_{i=1}^m u_i = k$. 
\begin{corollary}
\label{theorem:equality}
There is a $(13+\eps)$-approximation one-pass streaming algorithm for equality-based  fair $k$-centers   that stores  $O\left(km  \log(1/\eps) /\eps \right)$ points with running time $O(nk\log(1/\eps)/\eps)$. 
\end{corollary}

% === EXPERIMENTS === 

\section{Experiments}

\subsection{Set up}
For synthetic datasets, we generate twenty 4-dimensional Gaussian isotropic blobs with identity covariance matrix. We assign 5000 data points for each blob and thus 100000 points for the entire dataset.  Each blob's center is randomly initialized within a bounding box with edges of length 20. To create $m$ group assignments, we  generate   $\log_2 m$  random hyperplanes. Given a point $x$, we have a total of $2^{\log_2 m}$ possibilities of whether it lies above or below each hyperplane. Each corresponds to a group assignment. For real datasets, we use the Bank, Compas, and Adult datasets from the UCI  Repository \cite{Dua:2019}. The labels “deposit”, “sex”, and “race” are utilized to create group assignment respectively. We also normalize datasets and use numeric features for clustering. 

For all datasets, total number of centers $k$ is set as $5\%$ of total number of points. We are interested in a fair clustering where the number of points selected from each group is approximately proportional to the size of that group. As a result, we set $l_i = (1-\eps) \frac{\left\lvert S_i\right\rvert }{n} k,  \quad  u_i = (1+\eps) \frac{\left\lvert S_i\right\rvert}{n} k,$
where $\eps$ is a small value.

Previous approaches take in the required exact numbers of centers $m_i$ from each group as input and is only be applicable for perfectly fair clusterings where $\eps = 0$. Instead of setting $m_i = \frac{\left\lvert S_i\right\rvert}{n} k $, one could choose the values of $m_i$ in  $\left[ l_i, u_i\right]$ such that $\sum m_j = k$. We propose two different heuristics for choosing the values of $m_i$. Both approaches start by initializing $m_i = l_i$ then loops through each group and set $ m_i =  u_i$ if $ u_i - m_i \leq k - \sum m_i$, else $m_i = m_i + k - \sum m_j$.
%  $m_i = u_i$ if   $ u_i - m_i \leq k - \sum m_i$, and $m_i = m_i + k - \sum m_i$ otherwise. 
% \[ m_i = \begin{cases} u_i & \text{if } u_i - m_i \leq k - \sum m_i, \\
% m_i + k - \sum m_j & \text{otherwise.}
% \end{cases}
% \]
Intuitively, we set $m_i$ as the largest value possible if we have enough points left, $ u_i - m_i \leq k - \sum m_i$. Otherwise, we allocate  the remaining number of centers to $m_i$. The first (\textbf{Major}) and second (\textbf{Minor}) heuristics traverse the groups in  decreasing and increasing order of their size respectively. As $\eps$ increases,  one assigns more centers to groups that make up majority of the dataset, the other assigns more centers to minority groups.
%\textbf{Major} is motivated from the fact that large groups tend to spread out and could benefit from more centers.  \textbf{Major} is based from the fact that 
% \begin{table}[ht]
% \centering
% \small
% \tabcolsep=0.07cm

% \begin{tabular}{lcccc}
% \hline
% Algorithms & 50 Groups    & 100 Groups  & 200 Groups   & 400 Groups   \\  \hline
% Range $k$-centers   & \textbf{1.12 (0.04)}   & \textbf{1.44 (0.08)} & \textbf{2.05 (0.1)}   &\textbf{ 2.78 (0.21)}  \\
% Heuristic D  & 1.44 (0.09)  & 1.91 (0.14) & 2.48 (0.17)  & 3.5 (0.15)  \\
% Heuristic E & 1.55 (0.05)  & 1.84 (0.07) & 2.5 (0.12)  & 3.07 (0.16)  \\
%  \hline
% \end{tabular}
%  \caption{Mean and standard deviation of objective value on simulated data}
%  \label{table:synthPerformance}
% \end{table}

% \begin{figure}[htb]
%     \centering
%     \includegraphics[width=0.4\textwidth]{synth_range.png}
%     \caption{Mean runtime in seconds on simulated data}
%     \label{fig:runtime_synthetic}
% \end{figure}
For synthetic dataset, we have $m \in \{2,4,8\}$. For each dataset, we vary $\eps \in \{.1,.2,.3,.4\}$, report the mean objective value for each setting of 20 randomized runs in tables~\ref{table:synthPerformance} and~\ref{table:realPerformance}, and the corresponding standard deviations in tables ~\ref{table:synthPerformance_std} and~\ref{table:realPerformance_std}.
To complement our experiments, we also report mean objective values of the approaches in \cite{nguyen2020fair, kleindessner2019guarantees} with $\eps = 0$ in table~\ref{table:fairPerformance} and the mean runtime of the methods in figure~\ref{fig:runtime_real1}. 
\subsection{Results \& Discussion}
\textbf{Heuristic approaches are brittle and sensitive to values of} \textbf{$m_i$}.
In both settings, larger $\eps$ tends to worsen performance for the heuristic approaches. Modifying $m_i$'s is inadequate and ineffective to control the fairness-utility tradeoff.

\textbf{Our approach has superior performance for relatively small value of $\eps$}. As expected, for very small value of $\eps = .1$,  all the approaches have similar value compared to  the perfectly fair clusterings of \cite{nguyen2020fair} and \cite{ kleindessner2019guarantees}. For $\eps \geq .2$, our approach shows superior performance to the others. When $\eps = .2$, the gain in objective value is at least $10\%$ in both settings. As $\eps$ increases, the gain grows larger as expected and is up to at least $19\%$.     

\textbf{The runtime of our approach is comparable to the heuristic methods}. In figure~\ref{fig:runtime_real1}, we plot mean runtime of our approach to \cite{nguyen2020fair} and \cite{kleindessner2019guarantees} with the Minor heuristic. We omitted reporting the runtime of the Major heuristic methods since they have minimal difference with the Minor heuristics methods. The plot shows our algorithm is almost as fast as the heuristic methods.

% From table \ref{table:realPerformance} and figure \ref{fig:runtime_real1}, and \ref{fig:runtime_real2}, it's clear that our algorithm's performance gain is small but noticeable (around 6\%) across all datasets while running time is worse than Heuristic D.  

% \begin{table*}[htb]
% \centering
% \begin{tabular}{lccccc}
% \hline 
% Algorithms             & A-Gender & A-Race  & S-Sex & S-School & S-Address  \\ \hline
% Range $k$-centers  &\textbf{ 0.015 (0.00)}                &\textbf{ 0.015 (0.00)}  & \textbf{0.65 (0.001)} & \textbf{0.65 (0.002)} & \textbf{0.66 (0.002) }     \\
% Heuristic D             & 0.016 (0.00)  & 0.016(0.00) & 0.66 (0.006) & 0.7 (0.01)  & 0.7(0.01)     \\
% Heuristic E & 0.23 (0.001)  & 0.24 (0.002) & 0.7 (0.02) & 0.7 (0.03)    & 0.7(0.03)       \\
%  \hline
% \end{tabular}
%  \caption{Mean and standard deviation of objective value on real data}
%  \label{table:realPerformance}
% \end{table*}

\setlength{\tabcolsep}{4.25pt}
\begin{table*}[]

\footnotesize
\begin{tabular}{l|cccc|cccc|cccc}
\multicolumn{1}{c|}{\multirow{2}{*}{Algorithms }} & \multicolumn{4}{|c|}{\textbf{2 Groups}}                           & \multicolumn{4}{|c|}{\textbf{4 Groups}}                             & \multicolumn{4}{|c}{\textbf{8 Groups}}                            \\ 
\multicolumn{1}{c|}{}                                 & .1            & .2            & .3            & .4            & .1            & .2            & .3            & .4            & .1            & .2            & .3            & .4            \\ \hline
\textbf{Ours}                                         & 1.16                   & \textbf{0.93}          & \textbf{0.93}          & \textbf{0.94}          & \textbf{1.22}          & \textbf{1.11}          & \textbf{1.11}          & \textbf{1.11}          & \textbf{1.77}          & \textbf{1.44}          & \textbf{1.44}          & \textbf{1.44}          \\
(Jones et al.)-Minor                                     & 1.18                   & 1.22                   & 1.25                   & 1.33                   & 1.41                   & 1.49                   & 1.5                    & 1.58                   & 1.8                    & 1.86                   & 2.02                   & 2.14                   \\
(Jones et al.)-Major                                     & 1.2                    & 1.24                   & 1.25                   & 1.31                   & 1.42                   & 1.45                   & 1.53                   & 1.61                   & 1.83                   & 1.82                   & 1.95                   & 2.13                   \\
(Kleindessner et al.)-Minor                                     & \textbf{1.11}          & 1.14                   & 1.13                   & 1.2                    & 1.46                   & 1.5                    & 1.55                   & 1.62                   & 1.94                   & 2.02                   & 2.13                   & 2.15                   \\

(Kleindessner et al.)-Major                                     & 1.11                   & 1.13                   & 1.14                   & 1.2                    & 1.47                   & 1.5                    & 1.56                   & 1.61                   & 1.98                   & 2.08                   & 2.08                   & 2.17   
\end{tabular}
 \caption{Mean objective values on simulated data with varied $\eps$.}
 \label{table:synthPerformance}
\end{table*}

\begin{table*}[]
\footnotesize
\begin{tabular}{l|cccc|cccc|cccc}
\multicolumn{1}{c|}{\multirow{2}{*}{Algorithms}} & \multicolumn{4}{|c|}{\textbf{Compas}}                           & \multicolumn{4}{|c|}{\textbf{Bank}}                             & \multicolumn{4}{|c}{\textbf{Adult}}                            \\ 
\multicolumn{1}{c|}{}                                 & .1            & .2            & .3            & .4            & .1            & .2            & .3            & .4            & .1            & .2            & .3            & .4            \\ \hline
\textbf{Ours}                                        & .103          & \textbf{.091} & \textbf{.091} & \textbf{.088} & \textbf{.108} & \textbf{.107} & \textbf{.104} & \textbf{.105} & \textbf{.111} & \textbf{.108} & \textbf{.107} & \textbf{.107} \\
(Jones et al.)-Minor                                    & .105          & .115          & .125          & .130          & .114          & .120          & .126          & .133          & .130          & .133          & .141          & .147          \\
(Jones et al.)-Major                                    & .111          & .123          & .110          & .126          & .112          & .118          & .128          & .133          & .128          & .132          & .141          & .148          \\
(Kleindessner et al.)-Minor                                    & \textbf{.102} & .105          & .105          & .107          & .119          & .120          & .120          & .118          & .140          & .143          & .149          & .149          \\
(Kleindessner et al.)-Major                                    & .106          & .106          & .120          & .102          & .120          & .118          & .117          & .119          & .140          & .146          & .146          & .148         
\end{tabular}
 \caption{Mean objective values on real data with varied $\eps$.}
 \label{table:realPerformance}
\end{table*}
\begin{figure}[htb]
    \centering
    \includegraphics[width=0.5\textwidth]{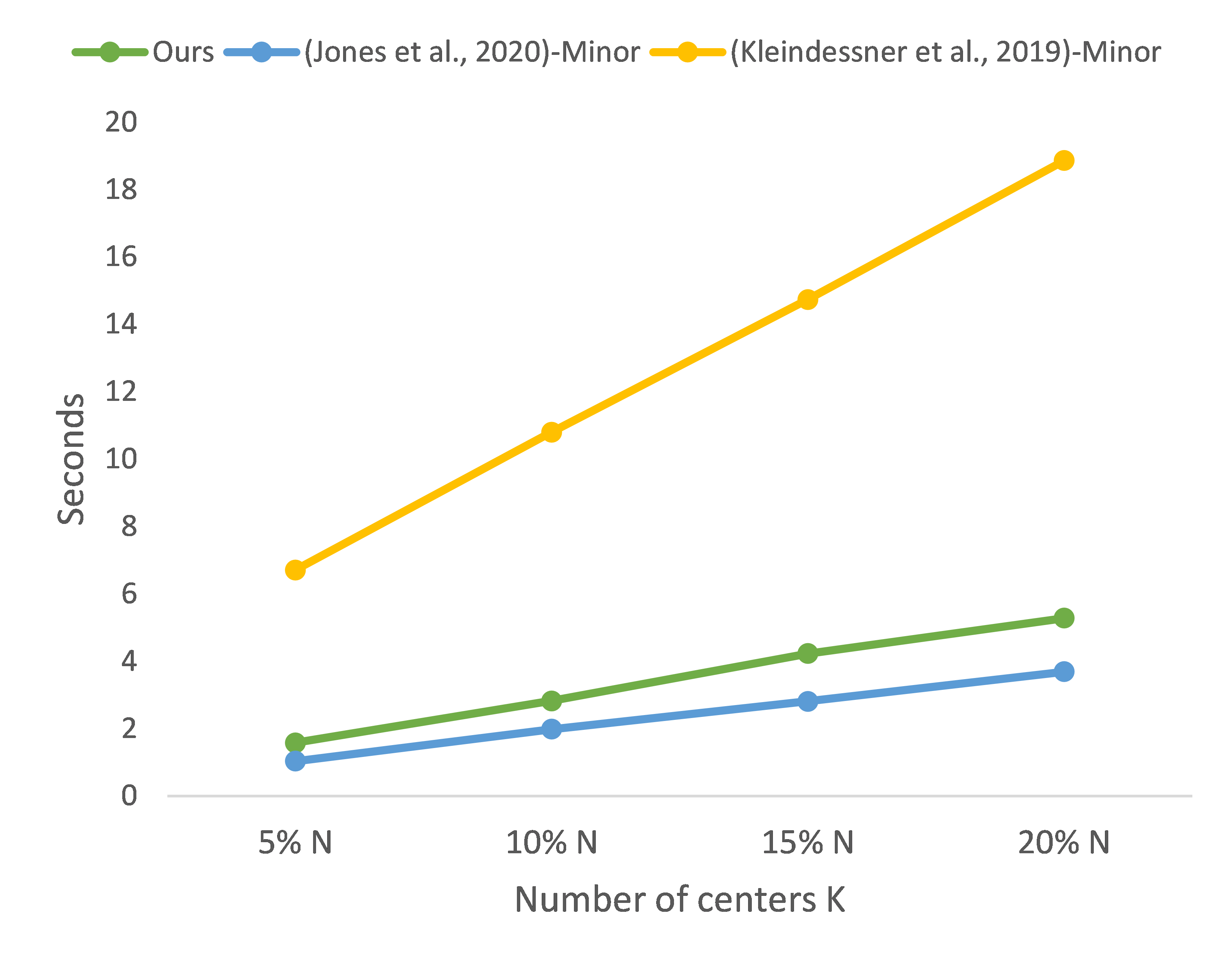}
    \caption{Mean runtime in seconds on Adult dataset.}
    \label{fig:runtime_real1}
\end{figure}
\begin{table}[]
\centering
\small
\begin{tabular}{l|ll}

Dataset  & (Jones et al.)  & (Kleindessner et al.)  \\ \hline
2 Groups  & 1.22 (.027)          & 1.24 (.100)         \\
4 Groups  & 1.4 (.039)           & 1.48 (.041)         \\
8 Groups  & 1.72 (.084) & 1.84 (.119)         \\
Compas   & .106 (.004)         & .017 (.011)       \\
Bank     & .113 (.004)         & .117 (.004)        \\
Adult    & .118 (.004)         & .121 (.009)      
\end{tabular}
 \caption{Means (and standard deviations) objective values and on synthetic and real data with $\eps = 0$.}
 \label{table:fairPerformance}
\end{table}

% \begin{figure}[htb]
%     \centering
%     \includegraphics[width=0.4\textwidth]{studentWholesale_range.png}
%     \caption{Mean runtime in seconds on student dataset}
%     \label{fig:runtime_real2}
% \end{figure}

\begin{table*}[htb!]

\footnotesize
\begin{tabular}{l|cccc|cccc|cccc}
\multicolumn{1}{c|}{\multirow{2}{*}{Algorithms }} & \multicolumn{4}{|c|}{\textbf{2 Groups}}                           & \multicolumn{4}{|c|}{\textbf{4 Groups}}                             & \multicolumn{4}{|c}{\textbf{8 Groups}}                            \\ 
\multicolumn{1}{c|}{}                                 & .1            & .2            & .3            & .4            & .1            & .2            & .3            & .4            & .1            & .2            & .3            & .4            \\ \hline
Ours & .044 & .001 & .001 & .001 & .115 & .002 & .001 & .002 & .058 & .002 & .002 & .002 \\
(Jones et al.)-Minor & .042 & .055 & .030 & .037 & .032 & .054 & .034 & .053 & .100 & .085 & .071 & .119 \\
(Jones et al.)-Major &  .055 & .050 & .035 & .016 & .078 & .059 & .056 & .032 & .122 & .110 & .106 & .141 \\
(Kleindessner et al.)-Minor &  .037 & .053 & .033 & .036 & .031 & .041 & .077 & .045 & .118 & .097 & .127 & .073 \\
(Kleindessner et al.)-Major &  .028 & .069 & .065 & .054 & .070 & .063 & .068 & .062 & .128 & .125 & .122 & .078
\end{tabular}
 \caption{Standard deviations of objective values on simulated data with varied $\eps$.}
 \label{table:synthPerformance_std}
\end{table*}

\begin{table*}[htb!]

\footnotesize
\begin{tabular}{l|cccc|cccc|cccc}
\multicolumn{1}{c|}{\multirow{2}{*}{Algorithms}} & \multicolumn{4}{|c|}{\textbf{Compas}}                           & \multicolumn{4}{|c|}{\textbf{Bank}}                             & \multicolumn{4}{|c}{\textbf{Adult}}                                                    \\ 
\multicolumn{1}{c|}{}                                 & .1            & .2            & .3            & .4            & .1            & .2            & .3            & .4            & .1            & .2            & .3            & .4            \\ \hline
Ours &  .694 & .718  & .256  & .385 & .413 & .435 & .588 & .469 & .408 & .038 & .035 & .029 \\
(Jones et al.)-Minor &  .422 & 1.002 & .987  & .920 & .359 & .502 & .691 & .521 & .468 &  .590 & .723 & .763 \\
(Jones et al.)-Major &  .596 & .552  & .836  & .736 & .542 & .557 & .576 & .626 & .710 & .552 & .803 & .478 \\
(Kleindessner et al.)-Minor &  .400 & 1.101 & 1.122 & .965 & .487 & .398 & .823 & .711 & .516 & .617 & .499 & .607 \\
(Kleindessner et al.)-Major & .734 & .167  & .681  & .701 & .558 & .469 & .663 & .605 & .658 & .714 & .953 & .599
\end{tabular}
 \caption{Standard deviations (scaled up by $100$) of objective values  on real data with varied $\eps$.}
 \label{table:realPerformance_std}
\end{table*}

\bibliography{ms}

\begin{thebibliography}{22}
\providecommand{\natexlab}[1]{#1}
\providecommand{\url}[1]{\texttt{#1}}
\expandafter\ifx\csname urlstyle\endcsname\relax
  \providecommand{\doi}[1]{doi: #1}\else
  \providecommand{\doi}{doi: \begingroup \urlstyle{rm}\Url}\fi

\bibitem[Adamczyk et~al.(2018)Adamczyk, Byrka, Marcinkowski, Meesum, and
  Wlodarczyk]{Adamczyk2018ConstantFF}
Marek Adamczyk, Jaroslaw Byrka, Jan Marcinkowski, Syed~Mohammad Meesum, and
  Michal Wlodarczyk.
\newblock Constant factor fpt approximation for capacitated k-median.
\newblock In \emph{ESA}, 2018.

\bibitem[Bera et~al.(2019)Bera, Chakrabarty, Flores, and
  Negahbani]{bera2019fair}
Suman Bera, Deeparnab Chakrabarty, Nicolas Flores, and Maryam Negahbani.
\newblock Fair algorithms for clustering.
\newblock In \emph{Advances in Neural Information Processing Systems}, pages
  4955--4966, 2019.

\bibitem[Bercea et~al.(2019)Bercea, Gro{\ss}, Khuller, Kumar, R{\"o}sner,
  Schmidt, and Schmidt]{Bercea2019OnTC}
I.~Bercea, M.~Gro{\ss}, S.~Khuller, Aounon Kumar, Clemens R{\"o}sner, Daniel~R.
  Schmidt, and M.~Schmidt.
\newblock On the cost of essentially fair clusterings.
\newblock \emph{ArXiv}, abs/1811.10319, 2019.

\bibitem[Celis et~al.(2018)Celis, Keswani, Straszak, Deshpande, Kathuria, and
  Vishnoi]{celis2018fair}
L~Elisa Celis, Vijay Keswani, Damian Straszak, Amit Deshpande, Tarun Kathuria,
  and Nisheeth~K Vishnoi.
\newblock Fair and diverse dpp-based data summarization.
\newblock \emph{arXiv preprint arXiv:1802.04023}, 2018.

\bibitem[Chakrabarty and Negahbani(2018)]{Chakrabarty2018GeneralizedCP}
Deeparnab Chakrabarty and Maryam Negahbani.
\newblock Generalized center problems with outliers.
\newblock In \emph{ICALP}, 2018.

\bibitem[Chen et~al.(2016)Chen, Li, Liang, and Wang]{chen2016matroid}
Danny~Z Chen, Jian Li, Hongyu Liang, and Haitao Wang.
\newblock Matroid and knapsack center problems.
\newblock \emph{Algorithmica}, 75\penalty0 (1):\penalty0 27--52, 2016.

\bibitem[Chen et~al.(2019)Chen, Fain, Lyu, and Munagala]{chenProportionally}
Xingyu Chen, Brandon Fain, Charles Lyu, and Kamesh Munagala.
\newblock Proportionally fair clustering.
\newblock \emph{CoRR}, abs/1905.03674, 2019.
\newblock URL \url{http://arxiv.org/abs/1905.03674}.

\bibitem[Chierichetti et~al.(2017)Chierichetti, Kumar, Lattanzi, and
  Vassilvitskii]{chierichetti2017fair}
Flavio Chierichetti, Ravi Kumar, Silvio Lattanzi, and Sergei Vassilvitskii.
\newblock Fair clustering through fairlets.
\newblock In \emph{Advances in Neural Information Processing Systems}, pages
  5029--5037, 2017.

\bibitem[Chiplunkar et~al.(2020)Chiplunkar, Kale, and Ramamoorthy]{CKR20}
Ashish Chiplunkar, Sagar Kale, and Sivaramakrishnan~Natarajan Ramamoorthy.
\newblock How to solve fair k-center in massive data models.
\newblock In \emph{Proceedings of the 37th International Conference on Machine
  Learning}, 2020.

\bibitem[Dua and Graff(2017)]{Dua:2019}
Dheeru Dua and Casey Graff.
\newblock {UCI} machine learning repository, 2017.
\newblock URL \url{http://archive.ics.uci.edu/ml}.

\bibitem[{EEOC} et~al.(1978){EEOC}, {CSC}, Justice, and Labor]{eeoc78}
{EEOC}, {CSC}, Justice, and Labor.
\newblock Uniform guidelines on employee selection procedures.
\newblock \emph{Federal Register}, 43\penalty0 (166):\penalty0 38290--38315, 8
  1978.

\bibitem[Gonzalez(1985)]{Gonzalez}
Teofilo~F Gonzalez.
\newblock Clustering to minimize the maximum intercluster distance.
\newblock \emph{Theoretical Computer Science}, 38:\penalty0 293--306, 1985.

\bibitem[Guha(2009)]{guha2009tight}
Sudipto Guha.
\newblock Tight results for clustering and summarizing data streams.
\newblock In \emph{Proceedings of the 12th International Conference on Database
  Theory}, pages 268--275, 2009.

\bibitem[Jones et~al.(2020)Jones, Nguy{\^e}n, and Nguyen]{nguyen2020fair}
Matthew Jones, Huy~L. Nguy{\^e}n, and Thy Nguyen.
\newblock Fair k-centers via maximum matching.
\newblock In \emph{Proceedings of the 37th International Conference on Machine
  Learning}, 2020.

\bibitem[Kale(2019)]{kale2019small}
Sagar Kale.
\newblock Small space stream summary for matroid center.
\newblock In \emph{Approximation, Randomization, and Combinatorial
  Optimization. Algorithms and Techniques (APPROX/RANDOM 2019)}. Schloss
  Dagstuhl-Leibniz-Zentrum fuer Informatik, 2019.

\bibitem[Khuller and Sussmann(2000)]{Khuller2000TheCK}
S.~Khuller and Y.~Sussmann.
\newblock The capacitated k-center problem.
\newblock \emph{SIAM J. Discret. Math.}, 13:\penalty0 403--418, 2000.

\bibitem[Kleinberg and Tardos(2005)]{algorithmDesign}
Jon Kleinberg and Eva Tardos.
\newblock \emph{Algorithm Design}.
\newblock Addison-Wesley Longman Publishing Co., Inc., USA, 2005.
\newblock ISBN 0321295358.

\bibitem[Kleindessner et~al.(2019{\natexlab{a}})Kleindessner, Awasthi, and
  Morgenstern]{FairkCCfDS}
Matth{\"a}us Kleindessner, Pranjal Awasthi, and Jamie Morgenstern.
\newblock Fair k-center clustering for data summarization.
\newblock \emph{arXiv preprint arXiv:1901.08628}, 2019{\natexlab{a}}.

\bibitem[Kleindessner et~al.(2019{\natexlab{b}})Kleindessner, Samadi, Awasthi,
  and Morgenstern]{kleindessner2019guarantees}
Matth{\"a}us Kleindessner, Samira Samadi, Pranjal Awasthi, and Jamie
  Morgenstern.
\newblock Guarantees for spectral clustering with fairness constraints.
\newblock In \emph{International Conference on Machine Learning}, pages
  3458--3467, 2019{\natexlab{b}}.

\bibitem[Li et~al.(2010)Li, Yi, and Zhang]{withDiversity2010}
Jian Li, Ke~Yi, and Qin Zhang.
\newblock Clustering with diversity.
\newblock \emph{CoRR}, abs/1004.2968, 2010.
\newblock URL \url{http://arxiv.org/abs/1004.2968}.

\bibitem[R{\"o}sner and Schmidt(2018)]{rosner2018privacy}
Clemens R{\"o}sner and Melanie Schmidt.
\newblock Privacy preserving clustering with constraints.
\newblock In \emph{45th International Colloquium on Automata, Languages, and
  Programming (ICALP 2018)}. Schloss Dagstuhl-Leibniz-Zentrum fuer Informatik,
  2018.

\bibitem[Schmidt et~al.(2020)Schmidt, Schwiegelshohn, and Sohler]{schmidt2020}
Melanie Schmidt, Chris Schwiegelshohn, and Christian Sohler.
\newblock Fair coresets and streaming algorithms for fair k-means.
\newblock In Evripidis Bampis and Nicole Megow, editors, \emph{Approximation
  and Online Algorithms}, pages 232--251, Cham, 2020. Springer International
  Publishing.
\newblock ISBN 978-3-030-39479-0.

\end{thebibliography}

%\newpage

%\section*{Appendix}
%\input{arxiv/appendix}
%\input{checklist}
\end{document}